\documentclass[12pt]{elsarticle}

\usepackage{amsfonts}
\usepackage{amsmath}

\newtheorem{theorem}{Theorem}
\newtheorem{definition}{Definition}
\newtheorem{lemma}{Lemma}
\newtheorem{remark}{Remark}
\newtheorem{corollary}{Corollary}
\newtheorem{proposition}{Proposition}

\newtheorem{proof}{Proof}

\def\Tr{\mathrm{Tr}}
\def\id{{\bf 1}\!\!{\rm I}}

\journal{ Chaos, Solitons \& Fractals}

\begin{document}

\begin{frontmatter}

\title{An Exponential Mixing Condition for  Quantum Channels}

\author{$^{a}$ Abdessatar Souissi}
\ead{a.souaissi@qu.edu.sa}

\author{$^{b}$ Abdessatar Barhoumi}
\ead{abarhoumi@kfu.edu.sa}

\affiliation{organization={Department of Management Information Systems, College of Business and Economics},
           addressline={ Qassim University},
          state={ Buraydah 51452},
          country={Saudi Arabia}}

\affiliation{organization={Department of Mathematics and Statistics, College of Science},
           addressline={King Faisal University},
           city={Al-Ahsa PO.Box: 400},
          state={31982},
          country={Saudi Arabia}}

\begin{abstract}
 Quantum channels, pivotal in information processing, describe transformations within quantum systems and enable secure communication and error correction. Ergodic and mixing properties elucidate their behavior. In this paper, we establish a sufficient condition for mixing based on a quantum Markov-Dobrushin inequality. We prove that if the Markov-Dobrushin constant of a quantum channel exceeds zero, it exhibits exponential mixing behavior. We explore limitations of some quantum channels, demonstrating that unistochastic channels are not mixing. Additionally, we analyze ergodicity of a class of mixed-unitary channels associated with finite groups of unitary operators. Finally, we apply our results to the qubit depolarizing channel.
\end{abstract}


\end{frontmatter}

\section{Introduction}
In quantum information processing, quantum channels, represented as completely positive trace-preserving maps, are essential for describing transformations within quantum systems \cite{CaGiLuMa14,  Qi22}. Beyond this, they are integral to quantum communication \cite{Chil2006}, like quantum teleportation and quantum key distribution \cite{Sca2009}, enabling secure transmission of quantum information. Quantum channels also play crucial roles in quantum error correction, quantum simulation . They serve as foundational tools across diverse applications in quantum information science.

The ergodic and mixing properties are of great interest for quantum systems. Ergodic and mixing conditions for quantum channels have been studied by several authors. In \cite{bau2007}, the authors investigate a Lyapunov type method to approach mixing property for quantum channels and discuss a case where mixing and ergodicity are equivalent. In \cite{MuWa2018}, the notion of S-mixing entropy of quantum channels have been introduced and investigated in connection with entanglement. In \cite{bau2013}, a systematic study of ergodicity and mixing in finite dimension have been provided. In \cite{AccLuSou22}, a quantum extension of the Markov-Dobrushin  inequality have been applied to quantum Markov operators and quantum channels.

A channel $\mathcal{M}$ is mixing if, for any pair of density operators $\rho$ and $\sigma$, repeatedly applying the channel to these densities eventually makes their difference vanish. In simpler terms, it means that after multiple transformations, the states become indistinguishable from each other. This implies, for a compactness reason, the existence of a unique fixed density operator $\rho_*$, which is the limit of the recursive transformation of an arbitrary starting input density $\rho$.

  A quantum channel $\mathcal{M}$, which operates on the space of all bounded linear operators   $\mathcal{B}(\mathcal{H})$ on a finite dimensional complex Hilbert space,  is considered ergodic if, when applied repeatedly to any input  density $\rho$, the average behavior of these iterated states converges consistently across the state space $\mathcal{B}(\mathcal{H})$. This convergence occurs regardless of the starting operator $\rho$, and it leads to a fixed density $\rho_*$, which is unique for the quantum channel $\mathcal{M}$. Ergodicity refers to the property that a system, over time, explores its entire state space and visits each part of it with sufficient frequency. It's important to recognize that, in most cases, ergodicity is considered a less stringent condition compared to mixing. This distinction was highlighted with a specific counter-example presented in  \cite{CaGiLuMa14}.

In the present paper, we prove a  sufficient condition for mixing for quantum channel, which is based on a quantum Markov-Dobrushin inequality \cite{AccLuSou22}. The main result of this research establishes a fundamental inequality governing the transformation of quantum states under the action of a quantum channel $\mathcal{M}$. Specifically, for any two initial quantum states $\rho$ and $\sigma$, the difference between their transformed states, denoted as $\mathcal{M}(\rho)$ and $\mathcal{M}(\sigma)$, respectively, is bounded by a factor dependent on the quantum Markov-Dobrushin constant $\kappa_{\mathcal{M}}$. Moreover, if $\kappa_{\mathcal{M}}$ exceeds zero, the quantum channel $\mathcal{M}$ exhibits a property known as 'mixing'. In this scenario, the convergence of $\mathcal{M}$ towards a unique fixed state $\rho_*$ is characterized by an exponentially rapid rate, governed by the parameter $\theta_{\mathcal{M}}$. Specifically, the rate of convergence is given by the expression: the difference between $\mathcal{M}^n(\rho)$ and $\rho_*$ is upper-bounded by $2e^{-n \theta_{\mathcal{M}}}$, where $\theta_{\mathcal{M}}$ is defined as $-\ln(1 - \mathrm{Tr}(\kappa_{\mathcal{M}}))$. This inequality sheds light on the behavior of quantum channels and their convergence properties, offering insights crucial for various quantum information processing tasks.  As a consequence, we show that  the quantum channel obtained by combining an of an arbitrary  channel with the completely depolarizing channel \cite{W2008} demonstrates exponential mixing behavior.

Moreover, we  explore the limitations of some quantum channels in mixing behavior. We show, for example, that unistochastic quantum channels, although bistochastic, do not display mixing behavior. Our investigation expands to the domain of mixed-unitary quantum channels, providing insights into their conduct of the mixing and convergence rate. We show that the average quantum channel associated with a finite group of unitary operators is ergodic if and only if it coincides with the completely depolarized channel.  To conclude, we investigate an explicit example of the qubit depolarizing channel and examine its  mixing and estimate its convergence rate.

Our approach is extendable to inhomogeneous quantum dynamics of quantum systems with promising implications in connections with quantum algorithms and quantum walks. An other interesting direction concerns the connection of mixing conditions for quantum channels and quantum Markov chains \cite{SSB23} and hidden quanutm Markov models \cite{SS23}. These problems  are of great interest, namely the Markov operators and quanutm channels are related in a duality relation \cite{AccOh99} that allow a quantum channels approach to the tudy of stochastic properties of quantum Markov processes. These problems  will be addressed in future works.

Let's briefly outline the structure of this paper. We begin with some preliminary notations in Section \ref{sect_prel}. Following that, our main focus is on proving the exponential mixing condition for quantum channels in Section \ref{sect_main}. In Section \ref{sect_Erg_mix}, we illustrate the implications of our results by examining the ergodicity and mixing properties of a specific class of mixed-unitary channels, supported by concrete examples. Finally, we conclude with some remarks in Section \ref{sect_concl}.

\section{Preliminaries}\label{sect_prel}

 Let $d\in\mathbb{N}$. Consider a  d-dimensional Hilbert space $\mathcal{H}$. By $\mathcal{B}(\mathcal{H})$, we denote the  algebra of (bounded) linear operators on $H$ with unity $\id$.

Denote $\mathfrak{S}(\mathcal{H})$ the subset of $\mathcal{B}(\mathcal{H})$   of all density operators. Notice that $\mathfrak{S}(\mathcal{H})$ is convex and compact.

 \begin{theorem}[{\cite{Choi75}}]
A super-operator $\mathcal{M}$ from  $\mathcal{B}(\mathcal{H})$ into itself is completely positive iff it can be expressed in terms  of a finite collection of operators $\{K_i\}_{i\in I}$ as follows
\begin{equation}\label{Krauss}
\mathcal{M}(A) = \sum_{i\in I} K_iAK_i^*,
\end{equation}
where $K_i$ are referred to as Kraus operators \cite{Kraus83}.
\end{theorem}


\begin{definition}
{\rm  A quantum channel is  defined to be a linear map that is completely positive and trace-preserving (CPTP), $\mathcal{M}$ from  $\mathcal{B}(\mathcal{H})$ into $\mathcal{B}(\mathcal{H})$.
}\end{definition}
The $n^{th}$ power of the quantum channel $\mathcal{M}$ corresponds to its composition  $n$ times
$$
\mathcal{M}^n = \mathcal{M}\circ \mathcal{M}\circ\cdots \circ \mathcal{M} $$
where $\mathcal{M}^0 = \mathcal{I}$ is the identity map on $\mathcal{B}(\mathcal{H})$.
\begin{remark}
Let $\mathcal{M} : \mathcal{B}(\mathcal{H}) \rightarrow \mathcal{B}(\mathcal{K})$ be a completely positive linear map. The Kraus decomposition for $\Phi$ can be represented as:
\begin{equation}
    \mathcal{M}(\rho) = \sum_{i} K_i \rho K_i^{*}
\end{equation}
where $K_i$ are the Kraus operators \cite{Kraus83}. It is worth noting that the condition for preserving the trace holds if and only if:
\begin{equation}
    \sum_{i} K_i^{*} K_i = \id
\end{equation}
\end{remark}

 \begin{definition}
   A quantum channel $\mathcal{M}$ acting on $\mathcal{B}(\mathcal{H})$ is said to be ergodic if, for any pair of density matrices $\rho \in \mathcal{B}(\mathcal{H})$, the time average of their iterates converges almost everywhere in the state space $\mathcal{B}(\mathcal{H})$:
\[
\lim_{n\to\infty} \frac{1}{n+1} \sum_{k=0}^{n} \mathcal{M}^k(\rho) = \rho_*
\]
where $\rho_*$ is the unique fixed point of the quantum channel $\mathcal{M}$, independent of the initial state $\rho$.
 \end{definition}

\begin{definition}
A quantum channel $\mathcal{M}$ is considered mixing if, for any pair of density matrices $\rho, \sigma \in \mathcal{B}(\mathcal{H})$, the difference between their iterates tends to zero as the number of iterations approaches infinity:
\begin{equation}\label{eq_df_mixing}
\lim_{n\to\infty} \|\mathcal{M}^n(\rho) - \mathcal{M}^n(\sigma)\| = 0
\end{equation}
where $\|\cdot\|$ denotes any norm defined on the (finite dimensional) algebra $\mathcal{B}(\mathcal{H})$.
\end{definition}

\begin{remark}
Quantum mixing implies that the influence of the initial states of the system decreases rapidly with time, causing the system to lose remnants of previous states These decaying interactions are usually quantified by the decay rate, indicating a rapid decrease in correlations.  Quantum mixing is thus considered more robust than ergodicity as it signifies a deeper level of randomness and disorder in the system, characterized by a rapid loss of correlation over time. However, it's important to note that the ergodicity of a quantum channel does not necessarily imply its mixing \cite{bau2007}.
 \end{remark}

\section{Main Results}\label{sect_main}
This section presents key results regarding the behavior and properties of quantum channels.   The main theorem provides a mixing condition for quantum channels, particularly focusing on their convergence rates using a quantum Markov-Dobrushin inequality.

Recall that, any operator $a\in\mathcal{B}(\mathcal{H})$ has a decomposition
$$
a = \frac{1}{2}(a + a^*) + i\frac{1}{2i}(a - a^*)
$$
into a real part $\Re(a) = \frac{1}{2}(a + a^*)$ and imaginary part $\Im(a) =\frac{1}{2i}(a - a^*)$ and both of them are real self-adjoint operators.
It is known that any real  self-adjoint  operator $b$ is the difference of two positive operators:
$$
b = b_+ -  b_-; \qquad b_+b_- = 0
$$
with disjoint support projections.

We define $\|b\|_1 = \Tr(b_+) + \Tr(b_{-})$
 In the upcoming lemma, we establish the  norm $\|\cdot \|$  in $\mathcal{M}_d(\mathbb{C})$ as an extension of the norm $\|\cdot\|_{1}$ to encompass the entire algebra. This norm plays a central role in the subsequent discussions. The lemma presented below has been  proved in \cite{AccLuSou22}.
\begin{lemma}{\rm
Define
\begin{equation}\label{df-q-tot-var-nrm}
\|A\|_{TV}
:= \text{Tr}\left(\Re(a)_+ +  \Re(a)_-\right)
+ \text{Tr}\left(\Im(a)_+ + \Im(a)_-\right)
\end{equation}
The expression $\| (\cdot) \|_{TV}$ thus defined serves as a norm on $\mathcal{B}(\mathcal{H})$, considering it as a
\textbf{real vector space}.
}\end{lemma}

    Recall that if $I$ is a set and $(b_{\alpha})_{\alpha\in I}$ is a family of positive operators such that $\|b_{\alpha}\|\le K$ for some $K\in\mathbb{R}_{+}$, we express
    \begin{equation}\label{definition-inf-pos-ops}         \inf\{b_{\alpha} \ : \ \alpha\in I\}         := \sup\{ b\in \mathcal{A} \ : \ 0\le b \le  b_{\alpha} \ , \ \forall \alpha\in I\}     \end{equation}

\begin{theorem}\label{thm_main}{\rm
    Let $\mathcal{M}: \mathcal{B}(\mathcal{H}) \to \mathcal{B}(\mathcal{H})$ be a quantum channel. Then:
    \begin{enumerate}
        \item For any $\rho, \sigma \in \mathfrak{S}(\mathcal{H})$, the following operator inequality holds:
        \begin{equation}\label{q-MD-ineq}
            \|\mathcal{M}(\rho) - \mathcal{M}(\sigma)\|_{\text{TV}} \leq \|\rho - \sigma\|_{\text{TV}}\big(1 - \operatorname{Tr}(\kappa_{\mathcal{M}})\big)
        \end{equation}
        where
        \begin{equation}\label{definition-kappa-q}
            \kappa_{\mathcal{M}} := \inf\big\{\mathcal{M}(|\xi \rangle\langle \xi|) : \xi\in\mathcal{H}; \;  \|\xi\| =1 \big\} \in \mathcal{B}(\mathcal{H})_+
        \end{equation}
        \item Moreover, if $\kappa_{\mathcal{M}} > 0$ then $\mathcal{M}$ is mixing, and its convergence to the unique fixed point $\rho_*$ on $\mathfrak{S}(\mathcal{H})$ is exponentially fast with convergence rate satisfying
        \begin{equation}\label{eq-cv-rate}
            \|\mathcal{M}^n(\rho) - \rho_{*}\| \leq 2e^{-n \theta_{\mathcal{M}; \, n}}
        \end{equation}
        for every $\rho \in \mathfrak{S}(\mathcal{H})$. Here $\theta_{\mathcal{M}}:= -\ln(1 - \Tr(\kappa_{\mathcal{M}}))$.
    \end{enumerate} }
\end{theorem}
\begin{proof}
  Let $\rho$ and $\sigma$ belong to the set of density matrices $\mathfrak{S}(\mathcal{H})$. Then, $\rho - \sigma$ can be expressed as a linear combination of rank-one projectors onto the basis elements of an orthonormal basis $\{ |\ell\rangle , 1\le \ell\le d \}$ of the Hilbert space $\mathcal{H}$. These projectors are associated with eigenvalues $\lambda_{\rho - \sigma; \ell}$, where $\lambda_{\rho - \sigma; \ell}$ denotes the eigenvalue corresponding to the eigenvector $|\ell\rangle$.
  \[
  \rho - \sigma = \sum_{\ell} \lambda_{\rho - \sigma; \ell} |\ell\rangle\langle \ell| = \sum_{\ell\in S_{+}(\rho - \sigma)}|\lambda_{\rho - \sigma; \ell}|\; |\ell\rangle\langle \ell| - \sum_{\ell\in S_{-}(\rho - \sigma)}|\lambda_{\rho - \sigma; \ell}| \; |\ell\rangle\langle \ell|
  \]
  Define sets $S_{\rho -\sigma; +}$ and $S_{\rho-\sigma; -}$ containing indices $\ell$ such that the associated eigenvalues are positive and negative, respectively.
  \[
  S_{\rho -\sigma; +}=  \{ \ell: \lambda_{\rho-\sigma; \ell}>0 \};\qquad S_{\rho-\sigma; -} = \{1,2,\dots, d\}\setminus S_{\rho-\sigma; +}
  \]
  We observe that the difference between the sums of positive and negative eigenvalues of $\rho - \sigma$ is zero.
  \[
  \sum_{\ell\in S_{\rho - \sigma; +}}|\lambda_{\rho - \sigma; \ell}|   - \sum_{\ell\in S_{\rho - \sigma; -}}|\lambda_{\rho - \sigma; \ell}|  = 0
  \]
  This leads to an expression for the total variation norm $\|\mathcal{M}(\rho-\sigma) \|_{TV}$ in terms of the eigenvalues and the quantum channel $\mathcal{M}$.
  \begin{align*}
  \|\mathcal{M}(\rho-\sigma) \|_{TV}
  &=  \Big\| \sum_{\ell\in S_{\rho-\sigma;+}}|\lambda_{\rho-\sigma;\ell}| \mathcal{M}(|\ell\rangle\langle \ell|)  -  \Big(\sum_{\ell\in S_{+}(\rho - \sigma)}|\lambda_{\rho - \sigma; \ell}|   - \sum_{\ell\in S_{-}(\rho - \sigma)}|\lambda_{\rho - \sigma; \ell}|\Big)\kappa_{\mathcal{M}}\\
  &-  \sum_{\ell\in S_{\rho-\sigma;-}}|\lambda_{\rho-\sigma;\ell}| \mathcal{M}(|\ell\rangle\langle \ell|)\Big\|_{TV}\\
  &=  \Big\| \sum_{\ell\in S_{\rho-\sigma;+}}|\lambda_{\rho-\sigma;\ell}| \Big(\mathcal{M}(|\ell\rangle\langle \ell|)  - \kappa_{\mathcal{M}}\Big) -  \sum_{\ell\in S_-}|\lambda_{\rho-\sigma;\ell}|\Big(\mathcal{M}(|\ell\rangle\langle \ell|)  - \kappa_{\mathcal{M}}\Big) \Big\|_{TV}\\
  &\le    \sum_{\ell\in S_{\rho-\sigma;+}}|\lambda_{\rho-\sigma;\ell}| \Big\|\mathcal{M}(|\ell\rangle\langle \ell|)  - \kappa_{\mathcal{M}}\Big\|_{TV} +  \sum_{\ell\in S_-}|\lambda_{\rho-\sigma;\ell}|\Big\|\mathcal{M}(|\ell\rangle\langle \ell|)  - \kappa_{\mathcal{M}}\Big\|_{TV}\\
  \end{align*}
  From the minimality of $\kappa_{\mathcal{M}}$, we get
  \[
  \Big\|\mathcal{M}(|\ell\rangle\langle \ell|)  - \kappa_{\mathcal{M}}\Big\|_{TV} = {\rm Tr}(\mathcal{M}(|\ell\rangle\langle \ell|) -\kappa_{\mathcal{M}}) \le 1- {\rm Tr}(\kappa_{\mathcal{M}})
  \]
  Utilizing properties of the total variation norm, we obtain
  \[
  \|\mathcal{M}(\rho)-\mathcal{M}(\sigma)\|_{TV} \le \Big(  \sum_{\ell\in S_{\rho-\sigma;+}}|\lambda_{\rho-\sigma;\ell}|    +  \sum_{\ell\in S_-}|\lambda_{\rho-\sigma;\ell}|\Big)(1- {\rm Tr}(\kappa)) = \|\rho -\sigma\|_{TV}\big(1- {\rm Tr}(\kappa)\big)
  \]
  which is  (\ref{q-MD-ineq}).
  For the second part of the theorem, one has
  \[
  \|\mathcal{M}^{n}(\rho)- \mathcal{M}^{n}(\sigma)\|  \le  (1 - \kappa_{\mathcal{M}})\|\mathcal{M}^{n-1}(\rho)- \mathcal{M}^{n-1}(\sigma)\|\le\cdots\le (1 - \kappa_{\mathcal{M}})^{n}\|\rho -\sigma\|
  \]
    It follows that
  \begin{equation}\label{eqMnsigmarho}
  \delta(\mathcal{M}^{n}(\mathcal{S}(\mathcal{H}))) := \sup_{(\rho,\sigma)\in \mathcal{S}(\mathcal{H})^2}\|\mathcal{M}^{n}(\rho)- \mathcal{M}^{n}(\sigma)\|  \le 2\big(1 - \kappa_{\mathcal{M}}\big)^{n}
  \end{equation}
  Then the diameters $\delta(\mathcal{M}^{n}(\mathcal{S}(\mathcal{H})))$ of the decreasing sequence of compact sets $\mathcal{M}^{n}(\mathcal{S}(\mathcal{H}))$ converge to $0$. It follows that there exists $\rho_*\in\mathcal{S}(\mathcal{H})$ satisfying
  \[
  \bigcap_{n\ge 0}\mathcal{M}^{n}(\mathcal{S}(\mathcal{H})) =\{\rho_*\}
  \]
  and $\rho_*$ is a fixed point of the quantum channel $\mathcal{M}$.
  Therefore, for every $\rho\in \mathcal{S}(\mathcal{H})$, the sequence $\mathcal{M}^{n}(\rho)$ converges to $\rho_*$.
  From (\ref{eqMnsigmarho}), we obtain (\ref{eq-cv-rate}).
\end{proof}
\begin{remark}
In the work by Accardi, Lu, and Souissi   \cite{AccLuSou22}, an operator-based adaptation of the Markov-Dobrushin inequality has been introduced, expanding upon the classical formulation. They also demonstrated its applicability to determining the convergence rate of certain classical Markov chains. Here, we further extend this convergence rate to encompass a broader class of quantum channels using equation (\ref{eq-cv-rate}).
\end{remark}
The completely depolarizing channel operates on $\mathcal{B}(\mathcal{H})$ and is defined as:
\begin{equation}\label{depola_channel}
\Omega(a) = \frac{\text{Tr}(a)}{d} \cdot \id
\end{equation}
The   channel $\Omega$ is mixing, its quantum Markov-Dobrushin constant is $\kappa_{\Omega} = \frac{1}{d}\cdot \id$ then $\Tr(\kappa_{\Omega}) = 1$ and its fixed point is $\rho_* = \frac{1}{d}\cdot \id$.
\begin{corollary}{\rm
  For any quantum channel $\mathcal{M}$ on $\mathcal{B}(\mathcal{H})$ and any $\alpha\in (0,1)$. The quantum channel
  \begin{equation}\label{MalphaOmega}
    \mathcal{M}_{\alpha} = \alpha \mathcal{M} + (1-\alpha)\Omega
  \end{equation}
  is exponentially mixing , where $\Omega$ is  the completely depolarizing channel defined in (\ref{depola_channel}). Furthermore, the convergence rate of $\mathcal{M}_{\alpha}$ satisfies $\theta_{\mathcal{M}_{\alpha}}\geq -\ln(\alpha)$ }.
\end{corollary}
 \begin{proof}
   Let $\rho\in\mathfrak{S}(\mathcal{H})$. We have
   $$
    \mathcal{M}_{\alpha}(\rho) = \alpha \mathcal{M}(\rho) + (1-\alpha)\Omega(\rho) = \alpha \mathcal{M}(\rho) + \frac{1-\alpha}{d}\cdot \id \ge \frac{1-\alpha}{d}\cdot \id
   $$
   Particularly,  the above inequality holds if $\rho$ is a rank-one projection. Then the Markov-Dobrushin constant of $\mathcal{M}_{\alpha}$ satisfies
   $$
   \kappa_{\mathcal{M}_{\alpha}}\ge  \frac{1-\alpha}{d}\cdot \id >0
   $$
   Thus, according to Theorem \ref{thm_main}, the channel $\mathcal{M}_{\alpha}$ exhibits exponential mixing. Applying the trace to the above  inequality, we get
   $$
   \Tr(\kappa_{\mathcal{M}_{\alpha}}) \ge 1-\alpha\quad  \Longrightarrow \quad  \theta_{\mathcal{M}_\alpha} = -\ln(1-\Tr(\kappa_{\mathcal{M}_{\alpha}}))\ge - \ln(\alpha)
   $$
   This finishes the proof.
 \end{proof}

\section{Ergodicity and mixing of mixed-unitary channels}\label{sect_Erg_mix}
In this section, we explore the ergodicity and mixing properties of mixed-unitary quantum channels. We begin by demonstrating a limitation in the case of unistochastic quantum channels, showing that they are not mixing. Then, we introduce mixed-unitary quantum channels, which are convex combinations of unistochastic channels. Following this, we investigate the ergodicity of average quantum channels associated with finite groups of unitary matrices. We prove that such channels are ergodic if and only if they coincide with completely depolarized channels. Lastly, we examine the mixing behavior of qubit channels, distinguishing between cases where the channels are mixing or stationary.
\subsection{Limitation in  the case of Unistochastic quantum channels}

\begin{proposition}{\rm
The unistochastic quantum channel associated with the unitary matrix \(U\), defined as:
\begin{equation}\label{eq:M-uni}
    \mathcal{M}_{U}(a) = UaU^*,
\end{equation}
is not mixing.}
\end{proposition}

\begin{proof}
For every \(a = \sum_{i,j}a_{ij}|i\rangle\langle j|\in\mathcal{B}(\mathcal{H})\), the action of \(\mathcal{M}_{U}(a)\) can be expressed as:
\[
    \mathcal{M}_{U}(a)= \sum_{i,j}\sum_{k,l}U_{ik}\overline{U_{jl}} a_{kl}|i\rangle\langle j|
\]

Let \(\xi = \sum_{i}\xi_{i} \, |i\rangle \in \mathcal{H}\) such that \(\|\xi\| = \sum_{i}|\xi_i|^2 = 1\). Since \(|\xi\rangle\langle\xi| = \sum_{kl}\xi_k\overline\xi_l|k\rangle\langle l|\), we have:
\begin{equation}\label{eq_Mu_exp}
    \mathcal{M}_{U}\Big(|\xi\rangle\langle\xi|\Big)= \sum_{i,j}\sum_{k,l}(U_{ik}\xi_k)(\overline{U_{jl} \xi_l})\, |i\rangle\langle j|
     = \Big|U\xi \rangle\langle U\xi\Big|
\end{equation}
If \(\xi\) and \(\nu\) are two orthogonal unit vectors from \(\mathcal{H}\), then \(U^n\xi\) and \(U^n\nu\) are also orthogonal. This implies that the projections \(\Big|U^n\xi \rangle\langle U^n\xi\Big|\) and \(\Big|U^n\nu \rangle\langle U^n\nu\Big|\) have disjoint supports. However,
\[
\Big\|\mathcal{M}^n_{U}(|\xi\rangle\langle \xi|) - \mathcal{M}^n_{U}(|\nu\rangle\langle \nu|)\Big\|_{TV} = \Big\| |U^n\xi\rangle\langle U^n\xi|  - |U^n\nu\rangle\langle U^n\nu| \Big\|_{TV} = 2
\]
which violates (\ref{eq_df_mixing}). This completes the proof.
\end{proof}
\begin{remark}
   Since  \(\mathcal{M}_{U}\) maps rank-one projections into themselves,  \(\kappa_{\mathcal{M}_{U}} = 0\). Hence, Theorem \ref{thm_main} cannot determine whether the unistochastic channel is mixing or not. However, in the above result we use an alternatively technique to show  that \(\mathcal{M}_U\) is indeed not mixing.
\end{remark}
\subsection{Mixed-unitary quantum channels}
Let $N$ be an integer. Let $U_1,U_2,\cdots, U_N$ be a collection  of  unitary operators acting on $\mathcal{H}$ and a probability distribution $p = (p_1,\dots, p_d)$ be  a probability . Put
\begin{equation}\label{df_mixUnit-Chann}
  \mathcal{M}(a) = \sum_{\ell=1}^{N}p_\ell U_iaU_\ell^{*};\quad a\in\mathcal{B}(\mathcal{H})
\end{equation}
Recall that, the map $\mathcal{M}$ defined by (\ref{df_mixUnit-Chann}) is called   \textit{mixed-unitary} channel \cite{ben2005}  on $\mathcal{B}(\mathcal{H})$. It is a convex combination of unistochastic channels. From (\ref{eq_Mu_exp}) we find
$$
\mathcal{M}(|\xi\rangle\langle \xi|) = \sum_{\ell=1}^{N}p_\ell |U_{\ell}\xi\rangle\langle U_{\ell}\xi |
$$
\begin{theorem}
Let $G$ be a finite group of unitary matrices. The average quantum channel associated with the group $G$ is defined as
\[
\mathcal{M}_{G}(a) = \frac{1}{|G|}\sum_{U\in G} U a U^*, \quad \text{for all } a\in\mathcal{B}(\mathcal{H}),
\]
The channel $\mathcal{M}_{G}$ is ergodic if and only if it coincides with the completely depolarized channel $\Omega$, defined by (\ref{depola_channel}).
\end{theorem}

\begin{proof}
The completely depolarizing channel $\Omega$ can be interpreted as a mixed unitary channel. Since $\Omega$ is mixing, it is also ergodic.

Conversely, $\mathcal{M}_{G}$ is a mixed-unitary channel with uniform probability $p_U = \frac{1}{|G|}$ for every $U\in G$. It can be observed that $\mathcal{M}_G^2 = \mathcal{M}_{G}$. Particularly, for every density operator $\rho$, we have
\[
\frac{1}{n+1}\sum_{k=0}^{n} \mathcal{M}_G^k(\rho) = \frac{n}{n+1}\mathcal{M}_G(\rho) + \frac{1}{n+1}\rho \longrightarrow \mathcal{M}_G(\rho) \quad \text{as } n\longrightarrow \infty.
\]
Hence, $\mathcal{M}_{G}$ is ergodic if and only if $\mathcal{M}_G(\rho) = \mathcal{M}_G(\sigma)$ for any $\rho, \sigma\in \mathfrak{S}(\mathcal{H})$, the space of density operators.

This implies that for every $\sigma\in\mathfrak{S}(\mathcal{H})$, we have
\[
\mathcal{M}_G(\rho) = \mathcal{M}_G\left(\frac{1}{d}\cdot \mathrm{id}\right) = \frac{1}{d}\cdot \mathrm{id} = \Omega(\rho),
\]
where $d$ is the dimension of $\mathcal{H}$. Therefore, $\mathcal{M}_{G}$ and $\Omega$ coincide on $\mathfrak{S}(\mathcal{H})$. Due to the linearity of $\mathcal{M}_G$ and the generic decomposition of any element $a\in\mathcal{B}(\mathcal{H})$ into positive matrices $a = \mathrm{Re}(a)_+ - \mathrm{Re}(a)_- + i \mathrm{Im}(a)_+ - i \mathrm{Im}(a)_-$, this coincidence extends to $\mathcal{B}(\mathcal{H})$.
\end{proof}

\subsection{Mixing of the Qubit Channel}

Let $\mathcal{H} = \mathbb{C}^2$, equipped with its canonical basis
$$
|1\rangle = \begin{pmatrix} 1 \\ 0 \end{pmatrix}\quad ; \quad |2\rangle = \begin{pmatrix} 0 \\ 1 \end{pmatrix}
$$
 Consider the Pauli spin matrices:
$$
\id = \begin{pmatrix} 1 & 0 \\ 0 & 1 \end{pmatrix}; \quad \sigma_x = \begin{pmatrix} 0 & 1 \\ 1 & 0 \end{pmatrix}; \quad \sigma_y = \begin{pmatrix} 0 & -i \\ i & 0 \end{pmatrix}; \quad \sigma_z = \begin{pmatrix} 1 & 0 \\ 0 & -1 \end{pmatrix}
$$

Let $\alpha\in (0,1)$. Define the map $\mathcal{M}(a) = \sum_{i=0}^{3}K_i^{*} a K_i$, where
$$
K_0 = \sqrt{1-\alpha}\; I, \quad K_1 = \sqrt{\frac{\alpha}{3}}\,\sigma_x, \quad K_2 = \sqrt{\frac{\alpha}{3}}\,\sigma_y, \quad K_3 = \sqrt{\frac{\alpha}{3}}\,\sigma_z
$$
From a geometric perspective, the depolarizing channel $\mathcal{M}$ can be visualized as a uniform contraction of the Bloch sphere, controlled by parameter $\alpha$. When $\alpha = 1$, the channel transforms any input state $\rho$ into the maximally-mixed state, resulting in a complete contraction of the Bloch sphere to the single-point $\frac{1}{2}\cdot\id$ at the origin.

Using the matrix representation $a = \begin{pmatrix} a_{11} & a_{12} \\ a_{21} & a_{22} \end{pmatrix}$ in the basis $\{|1\rangle, |2\rangle\}$, it is evident that $\mathcal{M}$ defines a quantum channel. The action of $\mathcal{M}$ on $a$ is given by:
$$
\mathcal{M}(a) = \begin{pmatrix} (1- \frac{2}{3}\alpha)a_{11} + \frac{2}{3}\alpha a_{22} & & (1- \frac{4}{3}\alpha)a_{12} \\
\\
 (1-\frac{4}{3}\alpha)a_{21} &  & (1- \frac{2}{3}\alpha)a_{22} + \frac{2}{3}\alpha a_{11} \end{pmatrix}
$$
The unique fixed density matrix of $\mathcal{M}$ is $\rho_{*} = \frac{1}{2}\cdot \id$.\\
The associated quantum Markov-Dobrushin constant is:
$$
\kappa_{\mathcal{M}} = \inf\left\{\mathcal{M}(|\xi\rangle\langle \xi|): \xi\in \mathbb{C}^2, \|\xi\| =1 \right\} = \frac{2\alpha}{3}\cdot \id
$$
Two cases are distinguished:
\begin{enumerate}
  \item For $\alpha\neq \frac{3}{4}$, the quantum channel $\mathcal{M}$ is mixing. For every initial state $\rho\in\mathfrak{S}(\mathcal{H})$, the dynamics $\mathcal{M}^n(\rho)$ converge to the unique fixed point $\rho_* = \frac{1}{2}$.
  \item For $\alpha = \frac{3}{4}$, every $\rho\in\mathfrak{S}(\mathcal{H})$ yields $\mathcal{M}(\rho) =\frac{1}{2}\cdot\id$. Thus, $\mathcal{M}$ is stationary. In this case the considered Qubit channel coincides with the completely depolarizing channel (\ref{depola_channel}) which is maximally mixed.
\end{enumerate}

\section{Conclusion}\label{sect_concl}
This paper delves into the behavior of quantum channels, focusing on ergodicity and mixing properties. Our main contribution lies in establishing a fundamental inequality governing the transformation of quantum states under the action of a quantum channel. By means of the quantum Markov-Dobrushin inequality, we provide a sufficient condition for mixing, revealing exponential convergence properties for channels with positive Markov-Dobrushin constants.

Moreover, we highlight the limitations of certain quantum channels. For instance, unistochastic channels, while bistochastic, do not display mixing behavior, as demonstrated through a counterexample. Furthermore, we explore mixed-unitary channels, offering insights into their behavior regarding mixing and convergence rates.

Our results have significant implications for quantum information processing tasks, enabling better understanding and utilization of quantum channels. Future research could extend these findings to more complex systems such as the inhomogeneous quantum dynamics in connection with complex quantum walks.

\section*{CRediT authorship contribution statement}  Abdessatar Souissi:Conceptualization, Investigation, Resources, Methodology, Writing original draft, Writing – review \& editing. Abdessatar Barhoumi: Visualization, Project administration,Methodology, Resources. \section*{Data availability} No data was used for the research described in the article
\section*{Conflicts of Interest}
The authors declare that they have no conflicts of interest.

\end{document}